\documentclass[10pt, journal]{IEEEtran}
\normalsize
\usepackage{amsthm,amssymb,amsmath,tikz,graphics,cite,float,graphicx,epstopdf,epsfig,verbatim,url,bbm,mathtools}
\usepackage{siunitx}
\usetikzlibrary{automata}
\usetikzlibrary{shapes,arrows}
\newtheorem{lem}{Lemma}
\usepackage[utf8]{inputenc}
\newtheorem{thm}{Theorem}

\newtheorem{mydef}{Definition}

\usepackage{capt-of}
\usepackage[noend]{algpseudocode}
\usepackage{algorithmicx}
\usepackage[ruled]{algorithm}
\allowdisplaybreaks

\tikzset{
  treenode/.style = {align=center, inner sep=0pt, text centered,
    font=\sffamily},
  arn_r/.style = {treenode, circle, black, draw=black, 
    text width=1.5em, very thick},
}
 
\begin{document}
\title{Low Complexity Secure Code (LCSC) Design for Big Data in Cloud Storage Systems}
\author{Mohsen~Karimzadeh~Kiskani$^{\dag}$,
         Hamid~R.~Sadjadpour$^{\dag}$
\thanks{M. K. Kiskani$^{\dag}$ and H. R. Sadjadpour$^{\dag}$ 
are with the Department of Electrical Engineering, University of California, Santa Cruz. Email: 
\{mohsen, hamid\}@soe.ucsc.edu.}, Mohammad~Reza~Rahimi$^{\ddag}$ and Fred Etemadieh$^{\ddag}$
\thanks{M. R. Rahimi$^{\ddag}$ and Fred Etemadieh$^{\ddag}$ are with Futurewei Technologies, Santa Clara, CA. Email: \{reza.rahimi, fred.etemadieh\}@huawei.com.}}
\maketitle \thispagestyle{empty}

\begin{abstract}
 In the era of big data, reducing the computational complexity of servers in data centers will be an  important goal. We propose Low Complexity Secure Codes (LCSCs) that are specifically designed to provide information theoretic security in cloud distributed storage systems. Unlike traditional coding schemes that are designed for error correction capabilities, these codes are only designed to provide security with low decoding complexity.  
These sparse codes are able to provide (asymptotic) perfect secrecy similar to  Shannon cipher. The simultaneous promise of low decoding complexity and perfect secrecy make these codes very desirable for cloud storage systems with large amount of data. The design is particularly suitable for large size archival data such as movies and pictures. The complexity of these codes are compared with traditional encryption techniques.
\end{abstract}


\IEEEpeerreviewmaketitle

\section{Introduction}
In the era of big data, it is becoming increasingly inefficient to apply traditional cryptographical algorithms to securely store data. The traditional methods of providing security require significant computational power and are not well optimized for big data applications.  

For instance, Hypertext Transfer Protocol Secure (HTTPS) protocol which is the backbone of internet security uses the Transport Layer Security (TLS) protocol stack in Transmission Control Protocol / Internet Protocol (TCP/IP) for secure and private data transfer. TLS is a protocol suite that uses a myriad of other protocols to guarantee security. Many of these sub-protocols consume a lot of CPU power and are complex processes which are not optimized for big data applications. 

For instance, TLS uses {\em public-key cryptography} paradigms to exchange the keys between the communicating parties through the {\em TLS handshake protocol}. One of the well-known key-exchange algorithms that is used in TLS handshaking protocol is the RSA algorithm. With the typical modular exponentiation algorithms used to implement the RSA algorithm, public-key operations take quadratic computations, private-key operations take cubic computations and key generation takes quartic computations with respect to the number of bits in the modulus. After the key-exchange, TLS uses {\em TLS record protocol} and algorithms like Advanced Encryption Standard (AES), which is the current adopted algorithm by the U.S. National Institute of Standards and Technology (NIST) for the encryption of electronic data, to form block ciphers for fixed block sizes of 128 bits. Each 128 bits of data therefore needs to undergo AES calculations to be transformed to ciphertexts. Albeit all the effort to enhance the performance of such algorithms, todays secure cryptographic protocols are not well suited in big data applications as they need to perform a significant number of computations. Such unnecessary CPU processing time and power renders cloud service providers to spend significant resources to maintain their secure cloud services.

On the other hand, most of current cryptographic algorithms are based on {\em computational security} paradigms. In computational security, it is inherently assumed that the {\em man-in-the-middle} is unable to perform complex computations and it does not have infinite processing time for cryptanalysis. Such algorithms are vulnerable to attacks in time and they may be broken by novel deciphering algorithms. For instance, Data Encryption Standard (DES) which, prior to AES, was the official Federal Information Processing Standard (FIPS) in the U.S. is no longer considered to be secure. 

With a focus on reducing the decryption complexity, we propose a completely new security paradigm for archival data in distributed cloud systems. Reduced decryption complexity ensures that a significant amount of redundant processing power for security purposes is avoided. This saves significant amount of resources for cloud service providers. Further, we prove that the proposed solution is {\em information theoretically secure} as opposed to computationally secure. Information theoretic secrecy guarantees that our solution is secure
regardless of the computational power of the adversary and the cryptanalysis time.  

Our method is based on random sparse codes that are specifically designed for security purposes. The sparsity of these codes allows the decryption which is done on the clouds to be a low complexity operation. Further, the specific design of our codes allows us to achieve security  using the randomness of the encryption operations. Our method is optimized to work with large number of files that need to be archived. We prove that larger number of files will result in more secure solutions. Our method asymptotically achieves perfect secrecy posed by Shannon in \cite{shannon1949communication}. In cases when the number of files is not very large, our method is still capable of achieving a level of obfuscation that is desirable for many applications such as storing private images and videos. 

The rest of the paper is organized as follows. Section \ref{sec_related} is dedicated to the related works in security of distributed storage systems and also the works on utilizing codes for security purposes. The assumptions and problem formulation are described in section \ref{sec_problem}. In section \ref{sec_dense_enc} we will prove that at least one dense encoding scheme exists that results in a secure solution and in section  \ref{sec_security} we will examine the security of our approach. 
The complexity analysis and simulation results are provided in section \ref{sec_sim} and the paper is concluded in section \ref{sec_conc}.

\section{Related Works}
\label{sec_related}

Secure transmission of data is usually implemented at higher layers of network. Recently, there is a significant interest in studying physical layer security. Physical layer security assumes that all receivers, both legitimate and eavesdropper, possess the same complete knowledge of the transmission technique. The main idea is based on the original idea that there is a transmitter (Alice) who wants to transmit data to a legitimate receiver (Bob) while an eavesdropper (Eve) tries to listen to this communication and obtain the information. In this scenario, Alice can adopt any kind of encoding, modulation or even randomization before transmission but both Bob and Eve are aware of the transmission technique being used by Alice. Therefore, if there is no noise or channel impairment, an eavesdropper can perfectly decode the message. Wyner in 1975 \cite{wyner1975wire} proved that in case of noisy wiretap channels, Alice can encode the message so that it reveals no information to the Eve. An important parameter in wiretap channel is the secrecy capacity. Secrecy capacity is defined as the highest transmission rate that data can be transmitted in a wiretap channel such that Eve cannot decode any information and its probability of error stays close to 0.5.

After the Wyner paper, many researchers started to study the wiretap channel capacity for different scenarios. Also achieving that capacity was another research objective in many publications.  Many researchers \cite{zhang2014polar, harrison2009physical,kiskani2016effect, kiskani2013social, kiskani2011novel, kiskani2010delay,vahidian2015relay, hosseini2016cloud} studied the use of error correcting codes for wiretap channels and other types of networks. 

In this paper, we are deviating from this common approach and investigating a completely different problem. Suppose Alice has some data that wants to store in a cloud storage system. However, Alice does not want the cloud storage system to be able to access this data. Therefore, if this data is accessed by Eve, no information can be obtained from the stored data. Further, when this information that is stored by Alice in the cloud is wiretapped by Eve during transmission, no useful information can be obtained by her. Note that we also assume that no encryption technique is used and only coding schemes are utilized to protect the data. There are some prior work that attempted to use existing coding schemes to address this issue. 

Some papers use fountain codes \cite{mackay2005fountain} for content retrieval. The advantages of coding in caching and storage systems have been shown in our previous works in \cite{kiskani2017throughput,kiskani2015application,kiskani2017secure1,kiskani2017secure2,kiskani2017secure3,kiskani2016capacity,kiskani2017multihop,kiskani2015opportunistic}. The application of fountain codes in distributed storage systems was also studied in  \cite{dimakis2006distributed}. Other types of erasure codes have been extensively used in storage systems. Maximum Distance Separable (MDS) codes are widely used in storage systems \cite{dimakis2010network,dimakis2011survey} due to their repair capabilities . However, certain requirements are needed to secure the applications that use these codes. Authors in \cite{dikaliotis2010security} also studied the security of distributed storage systems with MDS codes. Pawar et al.  \cite{pawar2010secure} studied the secrecy capacity of MDS codes. The authors in \cite{pawar2011securing1,pawar2011securing2} also proposed security measures for MDS coded storage systems. Shah et al. \cite{shah2011information} proposed information-theoretic secure regenerating codes for distributed storage systems. Rawat et al. \cite{rawat2014optimal} used Gabidulin codes on top of MDS codes to propose optimal locally repairable and secure codes for distributed storage systems. Unlike all of the references \cite{dikaliotis2010security,dimakis2011survey, dimakis2010network,pawar2010secure,dimakis2006distributed,rawat2014optimal,shah2011information,pawar2011securing1,pawar2011securing2}, this paper studies the use of sparse vectors to design codes to provide security for distributed storage systems. We will show that these codes can be effectively used to attain asymptotic perfect secrecy. 

Kumar et al. \cite{kumar2016secure} have proposed a construction for repairable and secure fountain codes. Reference \cite{kumar2016secure} achieves security by concatenating Gabidulin codes with Repairable Fountain Codes (RFC). Their specific design allows to use Locally Repairable Fountain Codes (LRFC) for secure repair of the lost data. Unlike \cite{kumar2016secure} which has focused on the security of the repair links using concatenated codes, the current paper provides security for the data storage by only using sparse vectors without any additional code usage such that perfect secrecy is achieved. 

Network coding schemes has been shown to be very efficient from a security point of view. Cai and Young \cite{cai2002secure} showed that network coding can be used to achieve perfect secrecy. Bhattad et al. \cite{bhattad2005weakly} studied the problem of ``weakly secure'' network coding schemes in which even without perfect secrecy, no meaningful information can be extracted from the network. Subsequent to \cite{bhattad2005weakly}, Kadhe et al. studied the problem of weakly secure storage systems in  \cite{ kadhe2014weakly1,kadhe2014weakly2}. Yan et al. also proposed \cite{yan2014weakly,yan2013algorithms} algorithms to achieve weak security and also studied weakly secure data exchange with generalized Reed Solomon codes. In our method, when using sparse vectors to design codes for cloud storage systems, the messages are encoded by combining them with each other to create the ciphertext. Hence, the ciphertext will not be independent of the message and the Shannon criteria may not be valid. Therefore it may be intuitive to think that these codes can only achieve weak security as opposed to perfect security. 
We will show that our unique code construction results in asymptotic perfect secrecy. 

As far as we know, in all previous works in literature, researchers use existing codes that were originally designed for error correction and used/modified it for security purposes. In this paper, we pose the following questions. ``{\em Can we design a code specifically for security of stored data in cloud storage systems?}'' More specifically, since we will face the problem of large quantities of data being generated in the cloud, ``{\em can we design a code that has significantly lower computational complexity for decoding as compared to commonly used encryption techniques?} The immediate benefit from such approach would be significant reduction on the number of servers needed to maintain contents securely in the cloud because of the reduction on the computational complexity associated with recovering the data. Note that the new code design may not have any error correcting capabilities. It means that commonly used codes such as MDS codes \cite{bhattad2005weakly} can be used to protect the data. Another benefit of this approach is that even the cloud storage provider is unable to access the contents which provides a level of privacy for users to store their data in the cloud. Clearly, if the cloud storage provider is unable to access the contents, any eavesdropper who is listening to the communication between the user (Alice) and the cloud (Bob), cannot obtain any useful information. 

Our final goal is to demonstrate that perfect secrecy as defined by Shannon in \cite{shannon1949communication} can be achieved without any need to store keys. Note that the original approach of Shannon cipher system \cite{shannon1949communication} is not practical since for each bit of information, we need to store one bit of key which practically doubles the storage capacity requirement. We demonstrate that without generating any key, it is feasible to achieve perfect secrecy asymptotically. Since in this work we generate equivalent of key by means of combining of contents together (encoding function), it is clear that this encoding function cannot be shared by Alice with the cloud or Eve. This is one fundamental difference between this work and common problem of wiretap channel. One can consider this as the private key that is generated by Alice to secure the contents and won't be shared with anyone.

\section{Problem Formulation}
\label{sec_problem}
Assume that a user wants to store files $f_1,f_2,\dots,f_m$ on a cloud system where each file has $Q$ bits, i.e. $f_i \in \mathbb{F}_{2^{Q}}$. We assume that these are archival data and extension of  this work to non-archival data remains as future work. The $m \times 1$ vector that represents all files is denoted by $\mathbf{f} = [f_1~f_2~\dots~f_m]^T$. 
\subsection{Encoding}
\label{sub_sec_enc}
The user encodes these files using an encoding matrix $\mathbf{A}$ of size $l \times m$ (where $l \ge m$ ) and creates\footnote{Assume that an encoding peice of software is running on the user which does all of this processing.} an encoded vector of $l$ coded files as $\mathbf{b} = \mathbf{A} \mathbf{f}$. Assume that the elements of $\mathbf{A}$ belong to the Galois Field $\mathbb{F}_2$. 
The user has a storage space of size $h << l$ and saves $h$ of these encoded contents locally and uploads the rest of them on the cloud. These $h$ encoded contents act similar to the key in traditional cryptography. Let $\mathbf{c}$ be a vector of size $(l-h) \times 1$ showing all the encoded contents stored on the cloud and $\mathbf{u}$ be a vector of size $h \times 1$ representing all the encoded contents saved on the user storage such that\footnote{Note that each row of vector $\mathbf{b}$ contains Q bits and for simplicity of presentation, we use vector representation.}
\begin{align}
 \mathbf{b} = \begin{bmatrix}
               \mathbf{c} \\
               \mathbf{u}
              \end{bmatrix}.
 \label{eqs_b_c_u}
\end{align}

\subsection{File Retrieval}
\label{sub_sec_dec}
To retrieve a content, we need to solve linear equations in Galois Field $\mathbb{F}_2$. The cloud contains $l-h$ coded files in $\mathbf{c}$ and the user stores $h$ coded files in $\mathbf{u}$ locally. When the legitimate user wants to download any of the contents, the application that runs by the user should solve the linear equation 
\begin{align}
 \mathbf{A} \mathbf{f} = \mathbf{b},
 \label{eq_linear_basic}
\end{align}
in Galois Field $\mathbb{F}_2$ to be able to respond to the download request by the user. Since $l >m$, the linear equation in \eqref{eq_linear_basic} has many solutions. Let the decoding matrix $\mathbf{D}$ be one of these solutions. This matrix can be split in two smaller matrices $\mathbf{D}_c$ of size $m \times (l-h)$ and $\mathbf{D}_u$ of size $m \times h$ denoting the cloud and user decoding matrices respectively. Therefore, the $m \times l$ decoding matrix $\mathbf{D}$ could be used to retrieve the files as 
\begin{align}
  \mathbf{f} = \mathbf{D} \mathbf{b} = \begin{bmatrix}
                                        \mathbf{D}_c ~ \mathbf{D}_u
                                       \end{bmatrix}
                                       \begin{bmatrix}
                                        \mathbf{c} \\ \mathbf{u}
                                       \end{bmatrix}
                                     =  \mathbf{D}_c \mathbf{c} + \mathbf{D}_u \mathbf{u}.
 \label{eq_linear_pseudo}
\end{align}
Notice that in equation \eqref{eq_linear_pseudo}, the first contributing term is computed on the cloud data and the second contributing term is calculated from the user stored data. It is important to notice that 
cloud only gets the matrix $\mathbf{D}_c$ and therefore it is not capable of decoding the data on its own. In the subsequent section we will show that the user contributing part $\mathbf{D}_u \mathbf{u}$ acts similar to the key and is crucial to secure decoding. 

\subsection{Low Complexity Secure Code (LCSC) Design}
\label{subsec_dense}
Our goal is to propose an encoding strategy such that $\mathbf{A}$ will be full rank with high probability and the rows of the matrix $\mathbf{A}$ are fairly dense. Further, we want to ideally have a relatively sparse decoding matrix $\mathbf{D}$. We will use the density of the encoding matrix $\mathbf{A}$ for security purposes and the sparsity of the decoding matrix $\mathbf{D}$ for low complexity decoding. In other words, we want the columns of the encoding matrix $\mathbf{D}$ to be sparse so that few cloud operations would be enough to retrieve a content while the rows of the matrix $\mathbf{A}$ are dense such that a relatively large number of files get encoded together to enhance the security of our system. 

From the definition of the decoding matric $\mathbf{D}$ we have  
\begin{align}
 \mathbf{D} \mathbf{A} = \mathbf{I}_m,
 \label{eq_A_D_prod}
\end{align}
where $\mathbf{I}_m$ is the  $m \times m$ identity matrix. If $\mathbf{a}_i$ is the $i^{th}$ column of  matrix $\mathbf{A}$, then equation \eqref{eq_A_D_prod} results in 
the following linear equation in $\mathbb{F}_2$
\begin{align}
 \mathbf{D} \mathbf{a}_i = \mathbf{e}_i,
 \label{eq_linear_eq_x_i}
\end{align}
where $\mathbf{e}_i$ is a vector of all zeros except at the $i^{th}$ position. For each column of matrix $\mathbf{A}$, e.g. $\mathbf{a}_i$, the {\em Hamming weight} of these vectors define the density of the encoding matrix $\mathbf{A}$. We will compute the density of $\mathbf{A}$ using the Hamming weight of all vectors $\mathbf{a}_i$.

In this paper, we start with a sparse decoding matrix $\mathbf{D}$ and we will show that an encoding matrix $\mathbf{A}$ exists such that the number of 1s in each row of the encoding matrix $\mathbf{A}$ is proportional to $\Theta(m)$. To prove this, let's denote   
\begin{align}
 \mathbf{A}=[\mathbf{a}_1,\mathbf{a}_2,\dots,\mathbf{a}_m],
 \label{eq_matrix_a}
\end{align}
where each $\mathbf{a}_i$ is an $l \times 1$ column vector. We will show that an encoding matrix $\mathbf{A}$ exists such that each $\mathbf{a}_i$ will have a Hamming weight of $\Theta(l)$. Notice that the vector $\mathbf{a}_i$ is a solution to the linear equation in \eqref{eq_linear_eq_x_i} in $\mathbb{F}_2$.

\begin{mydef}{\em 
 A random vector $\mathbf{w} = (\omega_1, \omega_2, \dots, \omega_m)^T \in \mathbb{F}_2^m$ is called {\em $\sigma$-sparse} if all of its elements are independent of each other and we have
    \begin{align}
  \mathbb{P}[\omega_i = 1 ] &= \frac{1}{2}(1-\sigma) \nonumber \\
  \mathbb{P}[\omega_i = 0 ] &= \frac{1}{2}(1+\sigma).
  \label{eq_ksparse_2}
 \end{align}
 }\label{def_sigma_sparse}
 \end{mydef}
 Assume that $l$ independent $\sigma$-sparse vectors $\mathbf{d}_1, \mathbf{d}_2, \dots, \mathbf{d}_l$ are the columns of the decoding matrix $\mathbf{D}$. In other words, let 
 \begin{align}
 \mathbf{D}=[\mathbf{d}_1,\mathbf{d}_2,\dots,\mathbf{d}_l].
 \label{eq_matrix_d}
\end{align}
 In the next section, we will compute the density of the encoding matrix $\mathbf{A}$.

\section{Existence of a Dense Encoding Matrix}
\label{sec_dense_enc}
In order to compute the density of the encoding matrix $\mathbf{A}$, we first prove the following useful lemmas.
 \begin{lem}{\em 
 For a vector $\mathbf{x} \in \mathbb{F}_2^l$ with Hamming weight $k$, we have 
 \begin{align}
  \mathbb{P}[\mathbf{D} \mathbf{x} = \mathbf{e}_i~ | ~wt(\mathbf{x})=k] = 2^{-m} \left(1 - \sigma^k \right) 
 \left(1 + \sigma^k \right)^{m-1}.
  \label{eqs_lem_fix_x}
 \end{align}
  }\label{lem_fix_x}
 \end{lem}
\begin{proof}
 Since the Hamming weight of $\mathbf{x}$ is equal to $k$, this means that $k$ vectors 
 from the set of all vectors $\mathbf{d}_1, \mathbf{d}_2, \dots, \mathbf{d}_l$ are added together
 to create $\mathbf{e}_i$. Let's denote these vectors by $\mathbf{d}_{e_1}, \mathbf{d}_{e_2}, \dots,
 \mathbf{d}_{e_k}$. Let $d^{e_j}_{i}$ denote the $i^{th}$ element of vector $\mathbf{d}_{e_j}$.
 Since the vectors $\mathbf{d}_{e_1}, \mathbf{d}_{e_2}, \dots, \mathbf{d}_{e_k}$ are 
 independent and their elements are also mutually independent, using binary summation over $\mathbb{F}_2$ we have 
 \begin{align}
  \mathbb{P}[\mathbf{D} \mathbf{x} = \mathbf{e}_i | wt(\mathbf{x})=k] & = 
  \mathbb{P}[\sum_{j=1}^k d^{e_j}_{i}= 1] \prod_{\substack{l'=1\\l' \neq i}}^m
  \mathbb{P}[\sum_{j=1}^k d^{e_j}_{l'} = 0].
  \label{eqs_lem_fix_x_proof}
 \end{align}
 We can easily prove that 
 \begin{align}
  \mathbb{P}[\sum_{j=1}^k d^{e_j}_{i} = 1] = \frac{1}{2}(1-\sigma^k) 
  \label{eqs_lem_fix_x_proof_2}
 \end{align}
To prove this, we can use induction on $k$. Equation \eqref{eq_ksparse_2} shows that it is valid 
for the base case $k=1$. Assume that it is valid for $k-1$. We have 
  \begin{align}
  &\mathbb{P}[\sum_{j=1}^k d^{e_j}_{i}= 1] =  
  \mathbb{P}[d^{e_k}_{i} = 1]\mathbb{P}[\sum_{j=1}^{k-1} d^{e_j}_{i} = 0] \nonumber \\
  &\qquad +   \mathbb{P}[d^{e_k}_{i} = 0]\mathbb{P}[\sum_{j=1}^{k-1} d^{e_j}_{i} = 1] 
  =\frac{1}{2}(1-\sigma) \frac{1}{2}(1+\sigma^{k-1}) \nonumber \\
  &\qquad +\frac{1}{2}(1+\sigma) \frac{1}{2}(1-\sigma^{k-1}) = \frac{1}{2}(1-\sigma^k)  
  \label{eqs_lem_fix_x_proof_22}
 \end{align}
 Similarly, using induction on $k$ and equation \eqref{eq_ksparse_2} we can also prove that 
 \begin{align}
  \mathbb{P}[\sum_{j=1}^k d^{e_j}_{l'} = 0] = \frac{1}{2}(1+\sigma^k).
 \end{align}
 Hence, equation \eqref{eqs_lem_fix_x_proof} can be simplified to 
  \begin{align}
  \mathbb{P}[\mathbf{D} \mathbf{x} = \mathbf{e}_i~ | ~wt(\mathbf{x})=k] = 2^{-m} \left(1 - \sigma^k \right) 
 \left(1 + \sigma^k \right)^{m-1} 
 \nonumber 
  \label{eqs_lem_fix_x_proof_final}
 \end{align}
 \end{proof}
  \begin{mydef}
  {\em 
    Let $\mathcal{F}^l \subseteq \mathbb{F}_2^l$ be the subset of all vectors in $\mathbb{F}_2^l$ with a Hamming weight of at least $\frac{l}{2}$.  
  }\label{def_set}
 \end{mydef}
 We will prove that with probability close to one a solution of equation \eqref{eq_linear_eq_x_i} belongs to $\mathcal{F}^l$. To find bounds on the probability that $\mathbf{e}_i$ is spanned by $\mathbf{d}_1, \mathbf{d}_2, \dots, \mathbf{d}_l$ 
we define a new random variable $Y_i$ and an indicator function $\mathbbm{1}_i(\mathbf{x}) $ as follows
\begin{mydef}
  {\em 
    Let $Y_i$ denote the number of vectors $\mathbf{x} \in \mathcal{F}^l$ such that $\mathbf{D} \mathbf{x} = \mathbf{e}_i$.
  }\label{def_Y_i}
\end{mydef}
\begin{mydef}
  {\em 
  Let $\mathbbm{1}_i(\mathbf{x}) $ be an indicator function which is equal to 1 if $\mathbf{D} \mathbf{x} = \mathbf{e}_i$ and equal to 0 otherwise.
  }\label{def_indicator}
\end{mydef}

\begin{lem}{\em
If $\mathbf{D} = [\mathbf{d}_1~ \mathbf{d}_2 ~\dots~ \mathbf{d}_l]$, then the average number of vectors $\mathbf{x} \in \mathcal{F}^l$ such that $\mathbf{D} \mathbf{x} = \mathbf{e}_i$ is equal to 
\begin{align}
 \mathbb{E}[Y_i] = 2^{-m}\sum_{j=\frac{l}{2}}^l \binom{l}{j} \left(1 - \sigma^j \right) 
 \left(1 + \sigma^j \right)^{m-1}.
\end{align}
}\label{lem_e_yi}
\end{lem}
\begin{proof}
  For every $\mathbf{x} \in \mathcal{F}^l$, we have  $ Y_i= \sum_{\mathbf{x} \in \mathcal{F}^l}  \mathbbm{1}_i(\mathbf{x})$ and $\mathbb{E}[Y_i] =  \sum_{\mathbf{x} \in \mathcal{F}^l}
\mathbb{P}[\mathbf{D} \mathbf{x} = \mathbf{e}_i]$. Using the result of Lemma \ref{lem_fix_x} and taking the summation over all values of $k$ proves the lemma. 
\end{proof}
\begin{lem}{\em
 We have, $\mathbb{E}[Y_i^2] \le   \mathbb{E}[Y_i] + \mathbb{E}^2[Y_i]$.
 }\label{lem_var}
\end{lem}
\begin{proof}
Since  $ Y_i= \sum_{\mathbf{x} \in \mathcal{F}^l}  \mathbbm{1}_i(\mathbf{x})$ we have 
 \begin{align}
  \mathbb{E}[Y_i^2] &= \mathbb{E} \left[ \sum_{\mathbf{x}_1 \in \mathcal{F}^l} \sum_{\mathbf{x}_2  \in \mathcal{F}^l}
  \mathbbm{1}_i(\mathbf{x}_1)   \mathbbm{1}_i(\mathbf{x}_2)  \right] \nonumber \\
  &= \sum_{\mathbf{x}_1 \in \mathcal{F}^l} \sum_{\mathbf{x}_2  \in \mathcal{F}^l} 
  \mathbb{E} [\mathbbm{1}_i(\mathbf{x}_1)   \mathbbm{1}_i(\mathbf{x}_2)] 
  = \sum_{\mathbf{x}_1 \in \mathcal{F}^l} \mathbb{P}[\mathbf{D} \mathbf{x}_1 = \mathbf{e}_i] \nonumber \\
  &+ \sum_{\mathbf{x}_1 \in \mathcal{F}^l} \sum_{\substack{\mathbf{x}_2  \in \mathcal{F}^l \\ 
  \mathbf{x}_1 \neq \mathbf{x}_2}} \mathbb{P}[\mathbf{D} \mathbf{x}_1 = \mathbf{e}_i] 
  \mathbb{P}[\mathbf{D} \mathbf{x}_2 = \mathbf{e}_i] = \mathbb{E}[Y_i]   \nonumber \\
  &+ \left(\mathbb{E}[Y_i]\right)^2  - \sum_{\mathbf{x}_1 \in \mathcal{F}^l} 
  \left(\mathbb{P}[\mathbf{D} \mathbf{x}_1 = \mathbf{e}_i] \right)^2
  \le \mathbb{E}[Y_i] + \mathbb{E}^2[Y_i]. \nonumber 
 \end{align}
\end{proof}
  \begin{lem}{\em
  The probability that $\mathbf{e}_i$ is the summation of at least $k$ vectors in $\mathbf{d}_1,\mathbf{d}_2,\dots,\mathbf{d}_l$ is lower bounded by 
\begin{align}
 \mathbb{P}[\exists \mathbf{x} \in \mathcal{F}^l 
 ~\textrm{s.t.}~  \mathbf{D} \mathbf{x} = \mathbf{e}_i] \ge \frac{1}{1+\frac{1}{\mathbb{E}[Y_i]}}.
\end{align}
  }\label{lem_lower}
 \end{lem}
\begin{proof}
Since $Y_i$ is a non-negative integer random variable, from the {\em second moment method} in 
probability theory we have 
 \begin{align}
  \mathbb{P}[Y_i > 0] 
 \ge \frac{\mathbb{E}^2[Y_i]}{\mathbb{E}[Y_i^2]}.
 \end{align}
Hence, using Lemma \ref{lem_var} we have 
 \begin{align}
  \mathbb{P}[\exists \mathbf{x} \in \mathcal{F}^l
 ~\textrm{s.t.}~  \mathbf{D} \mathbf{x} = \mathbf{e}_i] = \mathbb{P}[Y_i > 0] 
 \ge \frac{\mathbb{E}^2[Y_i]}{\mathbb{E}[Y_i^2]} \ge \frac{1}{1+\frac{1}{\mathbb{E}[Y_i]}} 
 \nonumber 
 \end{align}
\end{proof}
\begin{lem}{\em 
  For any $0 < \alpha < 1$, we have 
  \begin{align}
  \frac{1}{l+1} 2^{l H(\alpha)} \le  \binom{l}{\alpha l} \le 2^{l H(\alpha)},
  \end{align}
 where $H(\alpha)$ denotes the entropy, i.e., $H(\alpha) = -\alpha \log_2(\alpha) - (1 - \alpha) \log_2(1-\alpha)$.
}\label{lem_binom_entropy} 
\end{lem}
\begin{proof}
The proof can be found in the appendix of \cite{mackay1999good}.
\end{proof}
\begin{thm}{\em
 Let the vectors $\mathbf{d}_1, \mathbf{d}_2, \dots, \mathbf{d}_l$ be $\sigma$-sparse random vectors  belonging to $\mathbb{F}_2^m$ such that their average Hamming weight is asymptotically 
 non-zero, i.e. 
 \begin{align}
  \lim_{m \to \infty} \mathbb{E}[wt(\mathbf{d}_i)] = \frac{1}{2} \lim_{m \to \infty} m(1-\sigma) = c_1 > 0.
  \label{eqs_thm_sigma_sparse_0}
 \end{align}
 If $l =  m (1 + \epsilon)$ where $\epsilon > 0$ is an arbitrary constant with respect to  $m$, then with a probability close to one, at least one solution of equation \eqref{eq_linear_eq_x_i} belongs to $\mathcal{F}^l$ for large $m$.
 }\label{thm_ksparse}
\end{thm}
 \begin{proof}
  Consider a base vector $\mathbf{e}_i$ for $i=1,2,\dots,m$. Using Lemmas \ref{lem_e_yi} 
  and \ref{lem_binom_entropy} we have 
  \begin{align}
   \mathbb{E}[Y_i] &\ge 2^{-m} 
   \binom{l}{\frac{l}{2}}(1-\sigma^{\frac{l}{2}})(1+\sigma^{\frac{l}{2}})^{m-1} \nonumber \\
   &\ge \frac{2^{l-m} }{l+1} (1-\sigma^{\frac{l}{2}})(1+\sigma^{\frac{l}{2}})^{m-1} 
   \ge \frac{2^{l-m} }{l+1} (1-\sigma^{\frac{l}{2}}). \nonumber 
  \end{align}
  Since $\lim_{m \to \infty} \mathbb{E}[wt(\mathbf{d}_i)]  = c_1 > 0$ we have
 \begin{align}
  \lim_{m \to \infty} \sigma^{\frac{l}{2}} &= \lim_{m \to \infty} \left( 1 - 2 \frac{c_1}{m} \right)^{\frac{1}{2}m(1+\epsilon)}  = e^{-c_1 (1+\epsilon)} \nonumber 
 \end{align}
 Hence, for large $m$ we have 
 \begin{align}
  \mathbb{E}[Y_i] \ge 2^{m \epsilon} \frac{1}{m(1+\epsilon) + 1} 
  \left(1-e^{-c_1 (1+\epsilon)}\right). \nonumber 
 \end{align}
 Therefore,
 \begin{align}
 \lim_{m \to \infty} \frac{1}{\mathbb{E}[Y_i]} \le 
  &\lim_{m \to \infty} \dfrac{1}{2^{m \epsilon} \frac{1}{m(1+\epsilon) + 1} 
  \left(1-e^{-c_1 (1+\epsilon)}\right)} = 0  \nonumber 
 \end{align}
Using Lemma \ref{lem_lower} we have 
\begin{align}
  \lim_{m \to \infty} &\mathbb{P}[\exists \mathbf{x} \in  \mathcal{F}^l ~\textrm{s.t.}~ \mathbf{D}
  \mathbf{x} = \mathbf{e}_i]   \ge \lim_{m \to \infty} \dfrac{1}{1+\frac{1}{\mathbb{E}[Y_i]}} = 1. \nonumber 
\end{align}
This proves that with a probability approaching one, when $m \to \infty$ and $l=m(1+\epsilon)$ then at least one solution to equation \eqref{eq_linear_eq_x_i} belongs to $\mathcal{F}^l$.
 \end{proof}
The above argument proves that if $l = m(1+\epsilon)$ for some $\epsilon>0$, then at least one of the solutions of the equation \eqref{eq_linear_eq_x_i} will have a Hamming weight of $\frac{l}{2}$. Hence, for any $\sigma$-sparse decoding matrix $\mathbf{D}$, at least one encoding matrix $\mathbf{A}$ exists such that all of its columns will have a Hamming weight of $\frac{l}{2}$. Therefore, the average Hamming weight of the rows of matrix $\mathbf{A}$ will be equal to $\frac{m}{2}$ which means that at least one encoding matrix exists that on average has dense rows.

\section{Security}
\label{sec_security}
In this section, we will use the results of section \ref{sec_dense_enc} to achieve perfect secrecy. Let $\mathbf{A}$ be an $l \times m$ dense encoding matrix which has a sparse decoding matrix $\mathbf{D}$. As proved in section \ref{sec_dense_enc}, the rows of this matrix will have on average $\frac{m}{2}$ non-zero elements. Therefore, we can prove the following lemma, 
\begin{lem}
 {\em 
 If the number of non-zero elements of an encoding vector increases with the number of files $m$, then the asymptotic distribution of bits of the encoded files tend to uniform.
 }\label{lem_uniform_encoded}
\end{lem}
\begin{proof}
 This proof is given as Lemma 4 in \cite{kiskani2017secure3}.  
\end{proof}
We will now use equation \eqref{eq_linear_pseudo} to analyze the security of our approach. Let $\mathbf{f}_c = \mathbf{D}_c \mathbf{c}$ and $\mathbf{f}_u = \mathbf{D}_u \mathbf{u}$ denote the contributions from the cloud and the user for data retrieval. Then using module two addition we can re-arrange equation \eqref{eq_linear_pseudo} to
\begin{align}
 \mathbf{f}_c = \mathbf{f} + \mathbf{f}_u.
 \label{eq_shannon}
\end{align}
Cloud sends $\mathbf{f}_c$ from the cloud to the user and this information is subject to eavesdropping. Equation \eqref{eq_shannon} is similar to Shannon cipher problem \cite{shannon1949communication}. In Shannon cipher, an encoding function $\mathfrak{e}: \mathbb{M} \times \mathbb{K} \to \mathbb{C}$ is mapping a message $\mathfrak{M} \in \mathbb{M}$ and a key $\mathfrak{K} \in \mathbb{K}$ to a codeword $\mathfrak{C} \in \mathbb{C}$. In this problem, the message $\mathbf{f}$ is encoded by the key $\mathbf{f}_u$ and the ciphertext $\mathbf{f}_c$ is created and transmitted through the channel. We want to examine the criteria for achieving perfect secrecy using the above cipher. The following theorem provides the necessary and sufficient condition to obtain perfect  secrecy in Shannon cipher system.
\begin{thm}{\em 
 If $|\mathbb{M}|=|\mathbb{K}|=|\mathbb{C}|$, a coding scheme achieves perfect secrecy if and only if 
 \begin{itemize}
  \item For each pair $(\mathfrak{M}, \mathfrak{C}) 
  \in (\mathbb{M} \times \mathbb{C})$, there exists a unique key 
  $\mathfrak{K} \in \mathbb{K}$ such that $\mathfrak{C} = 
  \mathfrak{e}(\mathfrak{M},\mathfrak{K})$.
  \item The key $\mathfrak{K}$ is uniformly distributed in $\mathbb{K}$.
 \end{itemize}
 }\label{thm_shannon}
\end{thm}
\begin{proof}
 The proof can be found in section 3.1 of \cite{bloch2011physical}.  
\end{proof}
We will use Theorem \ref{thm_shannon} to prove that our approach can achieve asymptotic perfect secrecy. 

\begin{thm}{\em 
If $h$ grows with $m$ such that $m <2^{h}$, then the proposed encoding scheme provides asymptotic perfect secrecy against any eavesdropper wiretapping the communication between the cloud and the user.
}\label{thm_secrecy}
\end{thm}
\begin{proof}
We formulated this problem as a Shannon cipher system assuming that $\mathfrak{M}=\mathbf{f}$, $\mathfrak{K}=\mathbf{f}_u$, and $\mathfrak{C} = \mathbf{f}_c$. The condition $m < 2^{h}$ ensures that a unique key exists for each requested message. Therefore, for any pair $(\mathfrak{m}, \mathfrak{C}) \in (\mathbb{M}, \mathbb{C})$, a unique key $\mathfrak{K} \in \mathbb{K}$ exists such that $\mathfrak{C} = \mathfrak{m} + \mathfrak{K}$. Further, we are guaranteed to have $|\mathbb{M}|=|\mathbb{K}|=|\mathbb{C}|$.
Notice that the key  $\mathfrak{K}=\mathbf{f}_u$ belongs to the set of all possible bit strings with $Q$ bits. Lemma \ref{lem_uniform_encoded} proves that each encoded file is uniformly distributed among all $Q$-bit strings. Hence each key which is a unique summation of such encoded files is uniformly distributed among the set of all $Q$-bit strings. In other words, regardless of the distribution of the bits in files, $\mathbf{f}_u$ can be any bit string with equal probability for large values of $m$. Therefore, the conditions in Theorem \ref{thm_shannon} are met and perfect secrecy is achieved.
\end{proof}

\section{Complexity Analysis and Simulations}
\label{sec_sim}

In this section, we will compare the complexity of our proposed algorithm with the complexity of AES algorithm which is the standard cryptographic algorithm adopted by NIST in the U.S. and is a part of TLS and HTTPS protocols. AES is a block cipher with a block length of 128 bits. AES encryption consists of 10 rounds of processing for a 128-bit key with 128 bit block ciphers. \textcolor{black}{Initially, the key needs to undergo a KeyExpansion phase in which 10 different keys are created from the original key. Among many other operations, the KeyExpansion phase needs at least 50 byte XOR operations to create new keys. After this phase, the AES algorithm runs 10 rounds of operations each with a separate key that is generated in the KeyExpansion phase.} Each round of processing includes one single-byte based substitution step ({\em SubBytes}), a row-wise permutation step ({\em ShiftRows}), a column-wise mixing step ({\em MixColumns}),  and the addition of the round key ({\em AddRoundKey}). The order in which these four steps are executed is different for encryption and decryption. In the AddRoundKey step the input text is XORed with the key and the same thing happens during decryption. The goal of ShiftRows and MixColumns steps is to scramble the byte order inside each 128-bit block. All of these steps require a lot of operations which result in a large number of sequential operations either for decoding and encoding. Some of these operations are summarized in Table \ref{table_sim}. 

Table \ref{table_sim} compares the number of operations in our proposed algorithm with the number of operations in the AES algorithm to justify the significant improvement of our approach over the AES algorithm in terms of computational complexity. \textcolor{black}{As can be seen from this table, at least a total of 5268 bit XOR operations are performed on a 128 block. Therefore, the number of per-bit XOR operations will be equal to 41.125 bit XOR operations in AES  decryption algorithm. Notice that all of these steps in AES algorithm are done in sequential order and this will induce significant delays on AES encryption while out technique requires one-time XOR operation at the data center.}

Sparse decoding in our method allows us to perform the decoding with $m \frac{1}{2} (1-\sigma)$ XOR operations. Therefore for $m=128$, our approach requires $64(1-\sigma)$ XOR operations for decoding. This table clearly shows that our proposed algorithm significantly reduces the number of per-bit XOR operations from 41.125 to  4 XOR operations in case when $\sigma = 15/16$. It also does not require any other operations. \textcolor{black}{Also, our proposed decoding operation does not inflict significant delay on the system.}
\begin{table}[http]
    \centering
    \caption{Complexity comparison}
    \begin{tabular}{ |l|l|}
    \hline
    Algorithm & No. of operations\\
    \hline
    \hline
    KeyExpansion step in AES Algorithm & 40 byte row shifts\\
    & 50 byte table look-ups \\
    & 50 byte XOR operations\\
    \hline
    AddRoundKey step in AES Algorithm & 16 byte XOR operations\\
    \hline
    SubBytes step in AES Algorithm  & 16 byte table look-ups \\
    \hline 
    ShiftRows step in AES Algorithm  & 12 byte row shifts \\
    \hline 
    MixColumns step in AES Algorithm & 48 byte XOR operations\\
    & 64 byte table look-ups \\
    \hline
    Total number of byte XOR &\\
    operations in AES decryption& 658 byte XOR operations\\
    \hline
    Total number of bit XOR &\\
    operations in AES decryption& 5268 bit XOR operations\\
    \hline
    Total number of per-bit  XOR &\\  
    operations in AES & 41.125 XOR operations\\
    \hline
    Total number of per-bit XOR &\\ operations  in our
    sparse algorithm & \\ 
    with $\sigma = {15}/{16} = 0.9375$ & 4  XOR operations \\
    \hline
    \end{tabular}
    \label{table_sim}
\end{table}

Figure \ref{fig_avg_density} shows our simulation results. In this figure, we have plotted the average density of the encoding matrix for different number of files $m$ and different values of sparsity index $\sigma$. The case of $\sigma = 0.5$ is equivalent to random uniform decoding matrices which is not sparse. As we can see form this plot, when $\sigma = 0.5$ the average Hamming weight of the encoding matrix is equal to $\frac{m}{2}$ for different values of $m$. This is intuitively expected due to the symmetry and non-sparsity of the problem. When $\sigma$ becomes larger than $0.5$, the decoding matrix $\mathbf{D}$ becomes sparser which results in better decoding complexity and lower number of XOR operations. As can be seen from Figure \ref{fig_avg_density}, the average density of the encoding matrix becomes less than $\frac{m}{2}$ in this case. However, the simulation results confirm our theoretical results that regardless of the value of $\sigma$, i.e. regardless of the degree of sparsity of the decoding matrix when $m$ goes to infinity, the average Hamming weight of the encoding matrix approaches the dense value of $\frac{m}{2}$. Figure \ref{fig_avg_density} shows that for $\sigma = 0.6$ and $m>10$, or  when $\sigma = 0.7$ and $m>40$, or when $\sigma = 0.8$  and $m>60$, then the average density of the encoding matrix is close to $\frac{m}{2}$. Further, as can be seen from this plot, when $\sigma = 0.9$ the average density of the encoding matrix slowly approaches $\frac{m}{2}$ as $m$ goes to infinity.

\begin{figure}
    \center
      \includegraphics[scale=0.55,angle=0]{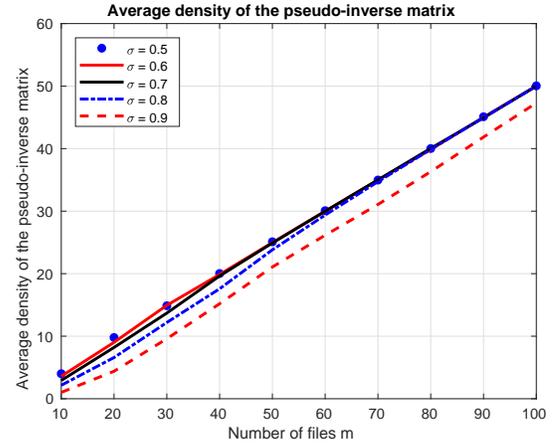}\\
      \caption{Average density of the encoding matrix versus the number of files for different sparse decoding matrices. 
      }
    \label{fig_avg_density}
\end{figure}

\section{Conclusions}
\label{sec_conc}
In this paper, we have proposed a radically different approach for providing secure low complexity storage solutions for cloud systems. Our method is specifically suited for big data applications when a large number of archival files needs to be stored on the cloud. We proved that our method is capable of achieving asymptotic perfect secrecy and through simulations and numerical studies  have shown that it is computationally  more efficient than today's cryptographic algorithms which are not well optimized to handle very large number of files.

We do not claim that this approach can replace encryption for all applications. For example, the decoding instruction that is exchanged between the user and cloud can be encrypted first before transmission. However, this approach can be a good alternative for archival data stored in cloud storage systems. Future  work will focus on studying the security properties of this approach for finite values of $m$. Extension of this work to other types of data is also desirable. This work was mainly focused on decoding properties of these codes but it may be useful to investigate low encoding complexity  secure codes. 

\section*{Acknowledgement}
The authors would like to acknowledge the generous support of Huawei Technologies in funding this research. This research is supported by Huawei Technology North America through award number TETF-9402566.

\bibliographystyle{unsrt}
\bibliography{conf1_bib} 
\end{document}